\documentclass[11pt]{article}
\usepackage{longtable,geometry}
\usepackage[latin1]{inputenc}
\usepackage{dsfont}
\usepackage{latexsym}
\usepackage{amsmath}
\usepackage{amsthm}
\usepackage{amssymb}

\usepackage{epsf}
\usepackage{epsfig}

\usepackage{enumerate}

\usepackage{multirow}
\usepackage{slashbox}

\newcommand{\Acurs}{\begin{cal}A\end{cal}}
\newcommand{\Ccurs}{\begin{cal}C\end{cal}}
\newcommand{\Ecurs}{\begin{cal}E\end{cal}}
\newcommand{\Fcurs}{\begin{cal}F\end{cal}}
\newcommand{\Hcurs}{\begin{cal}H\end{cal}}
\newcommand{\Lcurs}{\begin{cal}L\end{cal}}
\newcommand{\Mcurs}{\begin{cal}M\end{cal}}
\newcommand{\Pcurs}{\begin{cal}P\end{cal}}
\newcommand{\Ocurs}{\begin{cal}O\end{cal}}
\newcommand{\Scurs}{\begin{cal}S\end{cal}}
\newcommand{\Ucurs}{\begin{cal}U\end{cal}}
\newcommand{\Vcurs}{\begin{cal}V\end{cal}}

\newcommand{\Esp}{\mathbb{E}}
\newcommand{\F}{\mathbb{F}}
\newcommand{\LL}{{\mathrm{L}}}
\newcommand{\R}{\mathbb{R}}
\newcommand{\N}{\mathbb{N}}
\newcommand{\Proba}{\mathbb{P}}
\newcommand{\Ind}{\mathds{1}}
\newcommand{\card}{\mathrm{card}}
\newcommand{\Tr}{\mathrm{Tr}}

\newcommand{\esssup}{\mathop{\mathrm{ess\,sup}}}

\newcommand{\hypo}{{\mathrm{H}}}

\theoremstyle{plain}
\newtheorem{proposition}{Proposition}
\newtheorem{theorem}{Theorem}

\theoremstyle{definition}
\newtheorem{remark}{Remark}
\newtheorem{example}{Example}

\begin{document}

\title{Swing Options Valuation:\\a BSDE with Constrained Jumps Approach}

\author{Marie Bernhart\footnote{Laboratoire de Probabilit\'es et Mod\`eles Al\'eatoires, Universit\'es Paris 6-Paris 7, CNRS UMR 7599
		 and EDF R$\&$D, 92141 Clamart, France. Email: \texttt{marie-externe.bernhart@edf.fr}}
	\and Huyên Pham \footnote{Laboratoire de Probabilit\'es et Mod\`eles Al\'eatoires, Universit\'es Paris 6-Paris 7, CNRS UMR 7599, France. Email: \texttt{pham@math.jussieu.fr}}
	\and Peter Tankov\footnote{Centre de Math\'ematiques Appliqu\'ees, Ecole Polytechnique, 91128 Palaiseau, France. Email: \texttt{peter.tankov@polytechnique.org}}
  \and Xavier Warin\footnote{EDF R$\&$D, 92141 Clamart, France and Laboratoire de Finance des March\'es de l'Energie, Universit\'e Paris Dauphine. Email: \texttt{xavier.warin@edf.fr}}
  }

\maketitle

%
\abstract{
We introduce a new probabilistic method for solving a class of impulse control problems
based on their representations as Backward Stochastic Differential Equations (BSDEs for short) with constrained jumps.
As an example, our method is used for pricing Swing options. 
We deal with the jump constraint by a penalization procedure and apply a discrete-time backward scheme to the resulting penalized BSDE with jumps. 
We study the convergence of this numerical method, with respect to the main approximation parameters: the jump intensity $\lambda$, the penalization parameter $p > 0$ and the time step. In particular, we obtain a convergence rate of the error due to penalization of order $(\lambda p)^{\alpha - \frac{1}{2}}, \forall \alpha \in \left(0, \frac{1}{2}\right)$.
Combining this approach with Monte Carlo techniques, we then work out the valuation problem
of (normalized) Swing options in the Black and Scholes framework. We present numerical tests and compare our results with a classical iteration method. \\
}

\noindent
\textbf{Keywords} Backward stochastic differential equations with constrained jumps, Impulse control problems, Swing options, Monte Carlo methods


\newpage
%
\section{Introduction}
\label{sec:introduction}

In this report, we introduce a new probabilistic method for solving impulse control problems
based on their representations as Backward Stochastic Differential Equations (BSDEs for short) with constrained jumps.
As an example, our method is used for pricing Swing options in the Black and Scholes framework. \\

BSDEs provide alternative characterizations of the solution to multiple-obstacle,
optimal switching (see among others \cite{HJ07, CL07b, HT07, Por08, DH10}) and more generally impulse control problems:
Kharroubi et al. \cite{KM10} recently introduced a family of BSDEs with constrained jumps providing a representation
of the solution to such problems.
A challenging question is that of the numerical approximation of this kind of BSDEs with constrained jumps.

A discrete-time backward scheme for solving BSDEs with jumps (without constraint) has been introduced by Bouchard and Elie \cite{BouE08}.
In our case, the main difficulty comes from the constraint, which concerns the jump component of the solution.
These BSDEs do not a priori involve any Skorohod type minimality condition. In consequence,
classical approaches by projected schemes (discretely reflected backward schemes) used for example by \cite{BouC07} and \cite{CE10} cannot be used. \\

We consider a penalization procedure to deal with the constraint on jumps and provide a convergence rate of the penalized solution to the exact solution.
This allows us to establish a convergence rate of the error 
between the solution of the considered impulse control problem and the numerical approximation given by the discrete-time solution to the penalized BSDE with jumps,
as the penalization coefficient and the number of time steps go to infinity.

The rest of the report is structured as follows:
in Section \ref{sec:CJBSDE-representation}, we set the considered impulse control problem in the mathematical framework of BSDEs with constrained jumps.
We present in Section \ref{sec:CV-projection-scheme} our penalization approach and provide a global convergence rate of our approximation.
In Section \ref{sec:numerical-appli}, our method is used for pricing multi-exercise options, so-called (normalized) Swing options.
This multiple optimal stopping time problem leads to a particularly degenerate three-dimensional impulse control problem.
We combine our BSDE-based approach with Monte Carlo techniques and deal with Swing options with a small maximal number of exercises rights,
due to large computational times. We compare our pricing results with those obtained
by a classical iteration-based approach proposed for example by \cite{CT08}.

%
\section{BSDE Representation for Impulse Control Problems}
\label{sec:CJBSDE-representation}

Let $T$ be a given time horizon. We work in a complete probability space $(\Omega, \Fcurs, \Proba)$, on which is defined a $d$-dimensional Brownian motion $W$ and a Poisson process $N$ with intensity $\lambda > 0$. We denote by $\F = (\Fcurs_t)_{t \geq 0}$, the augmentation of the natural filtration generated by $W$ and $N$, by $\F^{W} = (\Fcurs^{W}_t)_{t \geq 0}$ the one generated by $W$, and by $\Pcurs$, the $\sigma$-algebra of predictable sub-sets of $\Omega \times [0, T]$. \\

\noindent\textbf{Notation} \\

Throughout this report, the euclidean norm defined on $\R^d$ or on $\R$ will be indiscriminately denoted by $| \cdot |$.
The matrix transposition is denoted by $\perp$. 
In addition, unless specified otherwise, $C$ will denote a strictly positive constant depending only on Lipschitz constants of the coefficients of the problem, see assumptions $(\hypo)$ and $(\hypo')$ below, and constants $T$, $| b(0) |$, $| \sigma(0) |$, $| \gamma(0) |$, $| f(0) |$, $| \kappa(0) |$ and $| g(0) |$. 

Besides, we shall use the standard notations:
\begin{itemize}
	\item $\Scurs^2$, the set of real-valued c\`adl\`ag adapted processes $Y = (Y_t)_{0 \leq t \leq T}$ such that
	$$
	\left\| Y \right\|_{\Scurs^2} := \left( \Esp \left[ \sup_{0 \leq t \leq T} \left| Y_t \right|^{2} \right] \right)^{\frac{1}{2}} < \infty \ .
	$$
	\item $\Acurs^{2}$, the sub-set of $\Scurs^2$ such that
	$$
	\Acurs^{2} := \left\{K \in \Scurs^2: (K_t)_{0 \leq t \leq T} \text{ nondecreasing , } K_0 = 0 \right\} \ .
	$$
	
	\item $\LL^{2}_{\F}([0, T])$, the set of real-valued adapted processes $(\phi_t)_{0 \leq t \leq T}$ such that
	$$
	\Esp \left[ \int_0^T \left| \phi_t \right|^{2} dt \right] < \infty \ .
	$$
	
	\item $\LL^2(W)$, the set of real-valued $\Pcurs$-measurable processes $Z = (Z_t)_{s \leq t \leq r}$ such that
	$$
	\left\| Z \right\|_{\LL^2(W)} := \left( \Esp \left[ \int_0^T \left| Z_t \right|^{2} dt \right] \right)^{\frac{1}{2}} < \infty \ .
	$$
	
	\item $\LL^2(N)$, the set of real-valued $\Pcurs$-measurable processes $V = (V_t)_{s \leq t \leq r}$ such that
	$$
	\left\| V \right\|_{\LL^2(N)} := \left( \Esp \left[ \int_0^T \left| V_t \right|^{2} \lambda dt \right] \right)^{\frac{1}{2}} < \infty \ .
	$$
	\item $\Vcurs$ denotes the set of $\Pcurs$-measurable essentially bounded processes, valued in $(0, \infty)$
	and $\Vcurs^{p} = \left\{ \nu^{p} \in \Vcurs: \nu^{p}_t \leq p \text{ a.s.}\right\}$.
\end{itemize}

\subsection{A Class of Impulse Control Problems} 

We consider the class of impulse control problems whose value function is defined by:
\begin{equation}
v(t, x) = \sup_{u = (\tau_k)_{k \geq 1} \in \Ucurs_{(t, T]}} \Esp \left[ g(X^{t, x, u}_T) + \int^{T}_t f(X^{t, x, u}_s) ds + \sum_{\stackrel{k \geq 1}{t < \tau_k \leq T}} \kappa(X^{t, x, u}_{\tau_k^{-}}) \right].
\label{general-impulse-control-pb}
\end{equation}
An impulse strategy $u = (\tau_k)_{k \geq 1}$ is said to be admissible for problem \eqref{general-impulse-control-pb} (and belongs to $\Ucurs_{(t, T]}$)
if it is a non-decreasing sequence of $\F^W$-stopping times valued in $(t, T]$ (we set by convention $\tau_0 = t$) such that, if
\begin{equation*}
n^u_{(t, T]} := \sharp \left\{ k \geq 1: t < \tau_k \leq T \right\}
\label{def-nb-impulse}
\end{equation*}
denotes the (random) number of interventions of the strategy $u$ before time $T$, then
\begin{equation}
\Esp \left| n^u_{(t, T]} \right|^{2} < C,
\label{est-nb-impulse}
\end{equation}
for some universal constant $C>0$. 
The controlled state variable $X^{t, x, u}$ is a c\`adl\`ag process with dynamics
\begin{equation}
		X^{t, x, u}_s = x + \int_t^s b(X^{t, x, u}_r) dr + \int_t^s \sigma(X^{t, x, u}_r) d W_r + \sum_{t < \tau_k \leq s} \gamma(X^{t, x, u}_{\tau_k^{-}}), \quad \forall s \geq t.
	\label{X-controlled-jump-EDS}
\end{equation}
Between two successive intervention times $\tau_{k}$ and $\tau_{k+1}$, the state variable evolves as a diffusion process
and the controller makes an integral profit $f$. At each decided intervention time $\tau_{k}$,
he gives an impulse to the system: the state process jumps with a size
$X^{u}_{\tau_{k}} - X^{u}_{\tau_{k}^{-}} = \gamma(X^{u}_{\tau_k^{-}})$ and he obtains the intervention gain $\kappa$. \\

We consider standard assumptions on the coefficients of the problem:
\begin{description}
	\item[$(\hypo)$] $b: \R^d \mapsto \R^d$, $\sigma: \R^d \mapsto \R^{d \times d}$ and $\gamma: \R^d \mapsto \R^d$ are Lipschitz continuous \\
									 \hspace*{0.4cm} and $\gamma$ is uniformly bounded. \\
		 							 \hspace*{0.4cm} $f: \R^d \mapsto \R$, $\kappa: \R^d \mapsto \R$ and $g: \R^d \mapsto \R$ are Lipschitz continous. \\
		 							 
	\item[$(\hypo')$] \hspace*{-0.2cm} The maps $b$, $\sigma$, $\gamma$, and $g$ belong to $\Ccurs^{1}_b(\R^d)$ and have Lipschitz continuous \\
									 \hspace*{0.4cm} derivatives.
\end{description}

A straightforward computation using $(\hypo)$, \eqref{est-nb-impulse} and Gronwall's lemma shows that
\begin{equation}
\forall (t, x) \in [0, T] \times \R^d, \quad \sup_{u \in \Ucurs_{(t, T]}} \Esp \left[ \sup_{t \leq s \leq T} \left| X^{t, x, u}_s \right|^{2} \right] < \infty.
\label{est-X-u}
\end{equation} 

Finally, we will assume the existence of an optimal strategy $u^{*} = (\tau^{*}_{k})_{k \geq 1} \in \Ucurs_{(t, T]}$ to problem \eqref{general-impulse-control-pb}.
We refer for example to \cite{BL84} and \cite{OS07} in the infinite horizon case,
for specific conditions on the coefficients of the problem which ensures such an existence. 

\subsection{Link to BSDEs with Constrained Jumps}

Let us consider the BSDE with constrained jumps
\begin{equation}
	\begin{cases}
		Y_t = g(X_T) + \int^{T}_{t} f(X_s) ds - \int^{T}_{t} Z_s dW_s - \int^{T}_{t} V_s dN_s + \int^{T}_{t} dK_s, \quad \forall 0 \leq t \leq T \\
 		V_t + \kappa(X_{t^{-}}) \leq 0, \quad \forall 0 \leq t \leq T
	\end{cases}
	\label{BSDE-const-jump}
\end{equation}
where $X$ is the (uncontrolled) jump diffusion process with dynamics
\begin{equation}
d X_t = b(X_t) dt + \sigma(X_t) d W_t + \gamma(X_{t^{-}}) dN_t.
\label{X-uncontrolled-jump-EDS}
\end{equation}
Under $(\hypo)$, this SDE admits an unique solution in $\Scurs^2$ and it is shown in \cite{KM10} that
under the additional assumption $(\hypo_1)$ given below, \eqref{BSDE-const-jump} admits a unique \textit{minimal} solution $(Y, Z, V, K) \in \Scurs^2 \times \LL^2 (W) \times \LL^2 (N) \times \Acurs^2$ with $K$ predictable.

The solution $(Y, Z, V, K)$ is said to be \textit{minimal} if and only if it has the smallest component $Y$ in the (infinite) class of solutions to \eqref{BSDE-const-jump}.
$(Y_t)_{t \geq 0}$ is called the value process and jumps with a size $V_t = Y_{t} - Y_{t^{-}}$.

\begin{description}
	\item[$(\hypo_1)$] There exists a solution $(\tilde Y, \tilde Z, \tilde K) \in \Scurs^2 \times \LL^2 (W) \times \Acurs^{2}$ to
	\begin{equation*}
	Y_t = g(X_T) + \int^{T}_{t} f(X_s) ds - \int^{T}_{t} Z_s dW_s + \int^{T}_{t} \kappa(X_{s^{-}}) dN_s + \int^{T}_{t} dK_s.
\end{equation*}
	\item[$(\hypo'_1)$] $(\hypo_1)$ holds and $\tilde Y_t = \tilde v(t, X_t), \forall 0 \leq t \leq T$ for some $\tilde v$ with linear growth. \\
	
	\item[$(\hypo_2)$] There exists a non negative function $\varphi \in \Ccurs^2(\R^d)$ and a constant $\rho>0$ s.t.
	\begin{equation*}
	\begin{array}{ll}
			\Lcurs \varphi + f \leq \rho \varphi, \quad & \varphi - \Hcurs \varphi > 0, \\
			\varphi \geq g, \quad & \lim_{|x| \rightarrow \infty} \frac{\varphi(x)}{1 + |x|} = \infty,
	\end{array}
	\end{equation*}
	in which $\Lcurs$ is the local component of the generator of the process $X$ and $\Hcurs$, the intervention operator:
\begin{align*}
\Lcurs v(t, x) &= b(x) \cdot D_x v(t, x) + \frac{1}{2} \Tr \left( \sigma \sigma^{\perp} (x) D^{2}_{x} v(t, x) \right), \\
\Hcurs v(t, x) &= v(t, x + \gamma(x)) + \kappa(x) \ .
\end{align*}
\end{description}

Let $(Y^{t, x}_s, Z^{t, x}_s, V^{t, x}_s, K^{t, x}_s)_{t \leq s \leq T}$ be
the solution to \eqref{BSDE-const-jump} when $X \equiv (X^{t, x}_s)_{t \leq s \leq T}$
is the solution starting at $x$ in $t$ to SDE \eqref{X-uncontrolled-jump-EDS}.
Under assumptions $(\hypo)$, $(\hypo'_1)$ and $(\hypo_2)$, \cite{KM10} show that the solution to impulse control problem \eqref{general-impulse-control-pb}
coincides with initial value of component $Y^{t, x}$:
\begin{eqnarray}
Y^{t, x}_t = v(t, x)
\end{eqnarray}
and is equal to the (unique) solution with linear growth to quasi-variationnal inequality
\begin{equation} \begin{array}{rlcl}
		\min \left\{ - \frac{\partial v}{\partial t}(t, x) - \Lcurs v(t, x) - f(t, x) \right. ; & & & \\
	 \left. v(t, x) - \Hcurs v(t, x) \right\} & = & 0, & \ \forall (t, x) \in [0, T) \times \R^d , \\
	 \min \left\{ v (T^{-}, x) - g(x) ; v(T^{-}, x) - \Hcurs v(T^{-}, x) \right\} & = & 0, & \ \forall x \in \R^d . \\
\end{array}
\label{QVI-relaxed}
\end{equation}
Let us mention that in general, the terminal condition $v(T^{-}, \cdot) = g$ is irrelevant, because of the possible discontinuity of $Y$ in $T^{-}$ due to constraints:
the relaxed terminal condition in \eqref{QVI-relaxed} expresses the possibility of a jump at time $T^{-}$. 

\begin{remark} For a better intuition, the following interpretation to solution $(Y, Z, V, K)$ holds when assuming $v \in \Ccurs^{1,2}([0, T], \R^d)$:
\[ \begin{array}{rrcl}
		\forall 0 \leq t \leq T, & Y_t & = & v(t, X_{t}) \\
		& Z_t & = & \sigma(t, X_{t^{-}}) D_x v(t, X_{t^{-}}) \\
		& V_t & = & v(t, X_{t^{-}} + \gamma(X_{t^{-}})) - v(t, X_{t^{-}}) \\
		&			 & = & \Hcurs v (t, X_{t^{-}}) - v(t, X_{t^{-}}) - \kappa(X_{t^{-}}) \\
		& K_t & = & \int_0^t \left( - \frac{\partial v}{\partial t}- \Lcurs v - f \right) (s, X_s) ds. \\
\end{array} \]
The constraint in \eqref{BSDE-const-jump} means thus that the obstacle condition is satisfied, namely $v(t, X_{t^{-}}) - \Hcurs v (t, X_{t^{-}}) \geq 0$.
\end{remark}

%
\section{Convergence of the Numerical Approximation by Penalization}
\label{sec:CV-projection-scheme}


It does not seem possible to use the minimality condition 
of the solution to BSDE with constrained jumps \eqref{BSDE-const-jump}
directly in a numerical scheme.
We thus propose an approach by penalization of the jump constraint.
The penalized constraint is introduced in the BSDE driver: when the constraint is fulfilled,
this penalization term disappears, and otherwise penalizes the driver with an exploding factor $p$.

In Theorem \ref{theo:cv-rate-global}, we provide an explicit rate of convergence of our approximation with respect to the parameters introduced:
namely, the jump intensity $\lambda$, the penalization coefficient $p$ and the time step.
Such an error estimate is essential for numerical purposes (understanding of the numerical impact of those parameters)
and allows to adjust in practice the fineness of the time grid in relation to $(\lambda, p)$.

\subsection{Approximation by Penalization} 

Given a parameter value $p>0$, the penalized BSDE is:
\begin{align}
	Y^{p}_t &= g(X_T) + \int^{T}_{t} \left[ f(X_s) + p \left( V^{p}_s + \kappa(X_{s^{-}}) \right)^{+} \lambda \right] ds \label{BSDE-jump-penalized} \\
					& \hspace*{1.5cm} - \int^{T}_{t} Z^{p}_s dW_s - \int^{T}_{t} V^{p}_s dN_s, \quad \forall 0 \leq t \leq T \nonumber
\end{align}
which admits an unique solution $(Y^{p}, Z^{p}, V^{p}) \in \Scurs^2 \times \LL^2 (W) \times \LL^2 (N)$ from the classical theory of BSDEs with jumps.
In addition, the sequence of penalized solutions $(Y^{p}, Z^{p}, V^{p})_{p}$ tends in $\LL^{2}_{\F}([0, T]) \times \LL^2 (W) \times \LL^2 (N)$
to the minimal solution $(Y, Z, V)$ to \eqref{BSDE-const-jump} as $p$ goes to infinity, see \cite{KM10}.
Besides, the convergence of $(Y^{p})_p$ to $Y$ is monotone and increasing. \\

Let $(Y^{p, t, x}, Z^{p, t, x}, V^{p, t, x})$ be the solution to \eqref{BSDE-jump-penalized} when $X \equiv (X^{t, x}_s)_{t \leq s \leq T}$.
We consider the following error introduced by this penalization procedure:
\begin{equation}
\Ecurs^{p} := \sup_{0 \leq t \leq T} \left| v(t, x) - Y^{p, t, x}_t \right|.
\label{def-penal-error}
\end{equation}

For any $t < \eta \leq T$, let us introduce 
\begin{eqnarray}
v^{\eta}_T(t, x) = \sup_{u = (\tau_k)_{k \geq 1} \in \Ucurs_{(t, T-\eta]}} \Esp \left[ g(X^{t, x, u}_{T}) + \int^{T}_t f(X^{t, x, u}_s) ds
																										+ \sum_{\stackrel{k \geq 1}{t< \tau_k \leq T}} \kappa(X^{t, x, u}_{\tau_k^{-}}) \right]
\label{problem-eta}
\end{eqnarray}
which corresponds to initial problem \eqref{general-impulse-control-pb} 
restricted to the sub-set of strategies taking values in $(t, T - \eta]$.
We shall denote by $u^{\eta *} = (\tau^{\eta *}_k)_{k \geq 1}$ an $\eta^{\frac{1}{2}}$-optimal
strategy to problem \eqref{problem-eta} (the existence of an optimal strategy is not ensured)
and by $n^{\eta *}$ be the number of impulses in strategy $u^{\eta *}$ that is:
$$
n^{\eta *} := \sharp \left\{ k \geq 1: t < \tau^{\eta *}_k \leq T - \eta \right\}.
$$
We will use the following additional assumptions: 
\begin{description}
	\item[$(\hypo^n)$] There exists some $\bar n \in \N^*$ such that
			$$
			\forall j \geq \bar n, \quad \Proba \left( n^{\eta *} \geq j \right) \leq l(j)
			$$
			for some map $l$ such that 
			$
			l(j) \leq e^{- C j} \text{ for some some constant $C > 0$}.
			$
	\item[$(\hypo^{*})$] There exists a map $h$ such that $h(\varepsilon) = \Ocurs_{\varepsilon \rightarrow 0} ( \varepsilon^{\frac{1}{2}} )$ and
	$$
	\forall \varepsilon > 0, \quad \Proba \left( \min_{k \geq 1} \left| \tau^{\eta *}_{k+1} - \tau^{\eta *}_k \right|\leq \varepsilon  \right) \leq  h (\varepsilon).
	$$
\end{description}

\begin{remark}[Assumptions $(\hypo^n)$ and $(\hypo^*)$] Both assumptions $(\hypo^n)$ and $(\hypo^*)$ are directly satisfied for the problem of Swing options valuation since 
the number of exercises right is almost surely bounded by some $n_{\max}$ and there is some fixed time delay $\delta > 0$ between two consecutive interventions.

More generally, $(\hypo^n)$ is intuitively satisfied as soon as the controlled state variable is constrained almost surely
and admits jumps of constant sign, 
see Example \ref{example:Hn}.
$(\hypo^*)$ is verified for sufficiently smooth problems, see for example the case of optimal forest management studied in \cite{Ber11}.
%
%
\end{remark}

\begin{example} Let us assume that the state variable defined in \eqref{X-controlled-jump-EDS} is such that
\begin{itemize}
	\item $b$ is uniformly bounded and $\sigma > 0$ constant,
	\item for some constant $c > 0$,
	$$
	\sup_{x \in \R^d} \gamma(x) \leq - c,
	$$
\end{itemize}
and that the optimal strategy $u^{*}$ implies
$
X^{u^*}_T \geq 0 \text{ a.s.}
$
Then a straightforward computation shows that the (random) number $n^{*}_{(0, T]}$ of optimal impulses before time $T$ satisfies, for any $a > 0$,
\begin{eqnarray*}
\Proba \left( n^{*}_{(0, T]} > n \right) = \Ocurs \left( e^{- an} \right) \quad \text{as } n \rightarrow + \infty.
\end{eqnarray*}
\label{example:Hn}
\end{example}

\begin{proposition} Let assumptions $(\hypo)$, $(\hypo^{n})$, $(\hypo^{*})$, $(\hypo'_1)$ and $(\hypo_2)$ be satisfied.
Then the penalization error in \eqref{def-penal-error} admits the following bound as $p$ goes to infinity:
$$
\Ecurs^{p} \leq C \left( \frac{\bar n {\bar C}^{\bar n}}{(\lambda p)^{\frac{1}{2} - \alpha}} \right), \quad \forall \alpha \in \left(0, \frac{1}{2}\right).
$$
for some constants $C > 0$ and $\bar C > 1$, which do not depend either on $\lambda$, $p$, $\bar n$ or $\alpha$.
\label{prop:cv-rate-penal}
\end{proposition}
\begin{proof} We provide the main arguments of the proof and refer the reader to \cite{Ber11} for more details.
%
%
%
The main idea comes from the following explicit functional representation available for $Y^{p, t, x}$, see \cite{KM10}:
\begin{equation}
Y^{p, t, x}_t = \esssup_{\nu^{p} \in \Vcurs^{p}} \Esp^{\nu^{p}} \left[ g(X^{t, x}_T) + \int^{T}_{t} f(X^{t, x}_s) ds + \int^{T}_{t} \kappa(X^{t, x}_{s^{-}}) dN_s \right],
\label{penalized-functionnal-rep}
\end{equation}
where $\Esp^{\nu^p}$ denotes the expectation under the probability measure $\Proba^{\nu^p}$ equivalent to $\Proba$ on $(\Omega, \Fcurs_T)$ with Radon-Nikodym density
\begin{equation*}
\left. \frac{d \Proba^{\nu^p}}{d \Proba} \right|_{\Fcurs_{T}}
= e^{- \int^{T}_{0} (\nu^p_s - 1) \lambda ds } e^{ \int^{T}_{0} \ln ( \nu^p_s ) dN_s }.
\end{equation*}
The specificity of such a change of measure is that it impacts only the jump parts of the processes:
under $\Proba^{\nu^p}$, the Brownian motion $W$ remains unchanged whereas $N$ has a (stochastic) intensity $(\lambda \nu^p_s)_{s \geq 0}$.
We shall denote by $N^{p}$, the doubly stochastic Poisson process (Cox process) with intensity $(\lambda \nu^{p}_s)_{s \geq 0}$ under $\Proba$,
by $(\tau^{p}_k)_{k \geq 1}$ the sequence of its jump dates and by $X^p$ the solution to \eqref{X-uncontrolled-jump-EDS} driven by $N^p$.

In view of \eqref{problem-eta} and \eqref{penalized-functionnal-rep}, 
we introduce 
a convenient measure change, which intuitively forces the penalized solution to jump as soon as possible after that an optimal impulse occurs:
\begin{equation}
	\forall s \geq 0, \hspace*{0.2cm} \nu^{p}_s =
		\begin{cases}
		p & \text{ if } \sum_{k \geq 1} \Ind_{\left\{ \tau^{\eta *}_k < s \right\}} \neq N^{p}_s, \\
		0 & \text{ else. }
		\end{cases}
		\label{measure-change}
\end{equation}
Notice that $\nu^p$ is a $\Pcurs$-measurable process bounded by $p$ a.s.
By definition of the counting process $N^p$, 
\begin{equation*}
\forall s \geq 0, \quad \Proba \left( \tau^{p}_k - \tau^{\eta *}_k > s \left. \right| \sum_{j \geq 1} \Ind_{\left\{ \tau^p_j \leq \tau^{\eta *}_k \right\}} = k-1 \right) = e^{- \lambda p s}.
\end{equation*}
In other words, the increment $\tau^{p}_k - \tau^{\eta *}_k$ has an exponential distribution
with parameter $(\lambda p)$, conditionally to the fact that 
$X^p$ has jumped one time less than $X^{u^{\eta*}}$.
This allows us to compute an estimate of the distance
$$
\sup_{0 \leq t \leq T} \left| v^{\eta}_T(t, x) - Y^{p, t, x}_t \right|,
$$
which holds for $\Ecurs^{p}$ by sending $\eta$ to $0$ together with a continuity argument
of the value function in its maturity variable.  
\qed
\end{proof}

\subsection{Convergence Rate of the Numerical Scheme} 

Given a regular time grid $\pi = \left\{ t_0 = 0, t_1, \ldots, t_N = T \right\}$,
we assume that the solution $X$ to \eqref{X-uncontrolled-jump-EDS} can be simulated on $\pi$
either perfectly or by using an Euler scheme and denote its discrete-time version by $X^{\pi}$.
%
Along the lines of \cite{BouE08}, we consider the following backward discrete-time scheme
for numerically solving the penalized BSDE with jumps \eqref{BSDE-jump-penalized}, 
\begin{equation}
	\begin{cases}
		Y^{p, \pi}_{t_N} = g(X^{\pi}_{t_N}) \\
		\forall t_n \in \pi, t_n < T : \\
		\hspace*{0.6cm} V^{p, \pi}_{t_n} = \frac{1}{\lambda \Delta t_{n+1}} \Esp \left[ Y^{p, \pi}_{t_{n+1}} \Delta \tilde N_{t_{n+1}} | \Fcurs_{t_n} \right] \\
		\hspace*{0.6cm} Z^{p, \pi}_{t_n} = \frac{1}{\Delta t_{n+1}} \Esp \left[ Y^{p, \pi}_{t_{n+1}} \Delta W_{t_{n+1}} | \Fcurs_{t_n} \right] \\
		\hspace*{0.6cm} Y^{p, \pi}_{t_n} = \Esp \left[ Y^{p, \pi}_{t_{n+1}} | \Fcurs_{t_n} \right] \\
		\hspace*{1.6cm} + \left[ f(X^{\pi}_{t_n}) + \left( p \left( V^{p, \pi}_{t_n} + \kappa(X^{\pi}_{t_n}) \right)^{+} - V^{p, \pi}_{t_n} \right) \lambda \right] \Delta t_{n+1} \\
	\end{cases}
\label{bwd-scheme-YZV}
\end{equation} 
where $\Delta t_{n+1} = t_{n+1} - t_{n}$, $\Delta W_{t_{n+1}}$ is the Brownian increment on $[t_n, t_{n+1}]$ and
$\Delta \tilde N_{t_{n+1}}$ the compensated version of the Poisson increment $\Delta N_{t_{n+1}}$ on $[t_n, t_{n+1})$. \\



We consider the classical discretization error between the continuous-time solution $(Y^p, Z^p, V^p)$ in \eqref{BSDE-jump-penalized} and its discrete-time approximation $(Y^{p, \pi}, Z^{p, \pi}, V^{p, \pi})$ in \eqref{bwd-scheme-YZV}, that is
\begin{align*}
\Ecurs^{\pi}(Y^{p}) & := \left( \max_{n<N-1} \sup_{t_n \leq t \leq t_{n+1}} \Esp \left| Y^{p}_t - Y^{p, \pi}_{t_n} \right|^{2} \right)^{\frac{1}{2}} \\
\Ecurs^{\pi}(Z^{p}) & := 
\left( \sum_{n = 0}^{N-1} \int_{t_n}^{t_{n+1}} \Esp \left| Z^{p}_t - Z^{p, \pi}_{t_n} \right|^{2} dt \right)^{\frac{1}{2}} \\
\Ecurs^{\pi}(V^{p}) & := 
\left( \sum_{n = 0}^{N-1} \int_{t_n}^{t_{n+1}} \Esp \left| V^{p}_t - V^{p, \pi}_{t_n} \right|^{2} \lambda dt \right)^{\frac{1}{2}}.
\end{align*}

Because of the lack of first order regularity of the driver of the penalized BSDE
\begin{equation}
f^p (x, v) := f(x) + \left( p (v + \kappa(x))^{+} - v \right) \lambda, \quad \forall (x, v) \in \R^d \times \R, 
\label{driver-YZV}
\end{equation}
classical regularization arguments for the FBSDE coefficients and Malliavin differentiation representations allow us to provide an explicit
convergence rate of order $|\pi|^{\frac{1}{2}}$ for errors $\Ecurs^{\pi}(Y^{p})$ and $\Ecurs^{\pi}(V^{p})$,
but only of order $|\pi|^{\frac{1}{4}}$ for error $\Ecurs^{\pi}(Z^{p})$, see Proposition \ref{prop:cv-error-discret}.

The impact of the penalization coefficient $p$ on the convergence of backward discrete-time schemes is well known in practice,
even if there does not exist, to our best knowledge, any explicit computation in the literature
(see for example the numerical experiments of \cite{Lem05} for the resolution by penalization of a BSDE with one reflecting barrier).
Basically, as $p$ increases at fixed discrete-time step, the quantity $f^{p} (\cdot) \Delta t_{n+1}$ explodes,
leading to a numerical explosion of the approximate values $Y^{p, \pi}$, see \eqref{bwd-scheme-YZV}.

We show rigorously in Proposition \ref{prop:cv-error-discret} that the discretization error grows exponentially
with $(\lambda p^2)$. This is due to the linear dependence in $\lambda$ and $p$ of the BSDE driver $f^p$, see \eqref{driver-YZV},
and estimate computations based on the use of Gronwall's lemma.

\begin{proposition} Assume $(\hypo)$. Then, as soon as
\begin{equation}
		|\pi| = \Ocurs \left( \frac{1}{\lambda p^2} \right),
\label{CN-cv-scheme}
\end{equation}
we get the following bounds as $p$ goes to infinity:
\begin{align*}
		\Ecurs^{\pi}(Y^p) & \leq C \left( (1 + \lambda)^2 \lambda p \bar C^{\lambda p^2} |\pi|^{\frac{1}{2}} \right) \\
		\Ecurs^{\pi}(V^p) & \leq C \left( (1 + \lambda)^2 \lambda^{\frac{3}{2}} p^{2} \bar C^{\lambda p^2}  |\pi|^{\frac{1}{2}} \right).
\end{align*}
Under $(\hypo')$, there exists a version of $Z^{p}$ such that, 
\begin{align*}
		\Ecurs^{\pi}(Z^p) & \leq C \left( (1 + \lambda)^2 \lambda^{\frac{5}{2}} p^{3} \bar C^{\lambda p^2} |\pi|^{\frac{1}{4}} \right),
\end{align*}
for some constants $C > 0$ and $\bar C > 1$, which do not depend either on $\lambda$, $p$ or $|\pi|$.
\label{prop:cv-error-discret}
\end{proposition}

\begin{proof} This follows from the same arguments as \cite{Eli08}: computations using It\^{o} and Gronwall's lemma and
regularization and Malliavin differentiation arguments applied to the penalized BSDE with jumps \eqref{BSDE-jump-penalized}.
We refer the reader to \cite{Ber11} for a detailed proof.
\qed
\end{proof}

Propositions \ref{prop:cv-rate-penal} and \ref{prop:cv-error-discret} enable us to establish a global convergence rate of the error introduced by our approximation by penalization.

\begin{theorem} Let the assumptions of Proposition \ref{prop:cv-rate-penal} be satisfied. Then
\begin{align}
		\Ecurs^{p}+ \Ecurs^{\pi}(Y^p) & 
		\leq C \left( \frac{1}{(\lambda p)^{\frac{1}{2} - \alpha}} + (1 + \lambda)^2 \lambda p \bar C^{\lambda p^2} |\pi|^{\frac{1}{2}} \right),
		\quad \forall \alpha \in \left(0, \frac{1}{2} \right) \label{cv-globale}
\end{align}
for some constants $C>0$ and $\bar C>1$, which do not depend either on $\lambda$, $p$, $|\pi|$, $\bar n$ or $\alpha$. Thus, for a sufficiently small time step $|\pi|$ with respect to $\lambda$ and $p$, the global error is such that
\begin{eqnarray*}
		\left[ \Ecurs^{p} + \Ecurs^{\pi}(Y^p) \right]^{*}
		&= \Ocurs \left( \frac{1}{(\lambda p)^{\frac{1}{2} - \alpha}} \right), \quad \forall \alpha \in \left(0, \frac{1}{2}\right).
\end{eqnarray*}

\label{theo:cv-rate-global}
\end{theorem}

\begin{remark}[Global convergence rate]
At fixed time step $|\pi|$, the convergence rate strongly deteriorates as $\lambda$ or $p$ increases, see \eqref{cv-globale}.
The numerical method is more sensitive to $p$ than to $\lambda$ according to \eqref{CN-cv-scheme} and \eqref{cv-globale}.
\eqref{CN-cv-scheme} constitutes a necessary condition for the convergence of the backward discrete-time scheme.
In practice, the penalization parameter will need to be chosen relatively small 
and the time step $|\pi|$ very small to avoid multiple jump times on each time step (otherwise, this introduces a bias).
\end{remark}

%
\section{Application to Swing Options Valuation}
\label{sec:numerical-appli}

In this section, our method is applied for the valuation of Swing options in the Black and Scholes framework.
This constitutes a multiple optimal stopping time problem which can be reformulated
as a particularly degenerate three-dimensional impulse control problem. 
We have been able to achieve convergence
for a small maximal number of exercises rights ($n_{\max} \leq 2$) due to the slow computational speed of our method.

\subsection{Swing Options Valuation as an Impulse Control Problem} 

We consider a (normalized) Swing option: the holder of the option is given a maximal number of exercise rights, say $n_{\max} \geq 1$, and has the opportunity to sell whenever he wants over a time period $[0, T]$ an underlying asset against a fixed strike price.

For some fixed strike price $K$, we shall denote by $\phi(s) = (K-s)^{+}$ the reward function corresponding to the profit made at each exercise date and by $S$ the underlying asset spot price. We concentrate here on the risk-neutral Black and Scholes framework in which $r > 0$ is a constant interest rate and the spot price process is defined by
\begin{equation}
S_t = s \exp \left\{ (r - \frac{1}{2} \sigma^{2}) t + \sigma W_t \right\}, \quad \forall t \geq 0,
\label{prix-BS}
\end{equation}
where $\sigma > 0$ denotes the volatility coefficient and $s$ an initial price value. \\

A delay $\delta>0$ between two consecutive exercise dates is introduced. Indeed, without any delay, the optimal strategy would consist in $n_{\max}$ simultaneous exercise at an unique optimal date (so that this option is equivalent to $n_{\max}$ identical American options).
%
The value of such an option at time $0$ can be written as the solution to the following multiple optimal stopping time problem
\begin{eqnarray}
\bar v^{(n_{\max})}(s) = \sup_{u = (\tau_k)_{k \geq 1} \in \Ucurs^{\delta, n_{\max}}_{(0, T]}} \Esp \left[ \sum_{k \geq 1} e^{- r \tau_{k}} \phi( S_{\tau_{k}} ) \right]
\label{v-Swing}
\end{eqnarray}
in which a strategy $u = (\tau_k)_{k \geq 1}$
is said to be admissible and belongs to $\Ucurs^{\delta}_{(0, T]}$
if and only if it is a vector of $\F^{W}$-stopping times valued in $(0, T]$ with maximal lenght $n_{\max}$ 
which satisfied the constraint on delay, that is
	\[ \begin{array}{rcl}
			\forall k \geq 1, \quad \tau_{k+1} - \tau_{k} \geq \delta. 
	\end{array} \]

As a multiple optimal stopping time problem, this problem can be formulated as a particular (and strongly degenerate) impulse control problem.
An impulse control corresponds to a sequence of exercise dates and the
intervention gain is written as the payoff function $\phi$ multiplied by an indicator
function which allows to satisfy both the constraint on the number of exercise rights and the constraint on delay between exercise dates. Namely:
\begin{align}
v(s) = \sup_{u = (\tau_k)_{k \geq 1} } \Esp \left[ \sum_{\stackrel{k \geq 1}{\tau_k \leq T}} e^{- r \tau_{k}} \phi( S_{\tau_{k}} ) \Ind_{ \left\{ \big( \Theta^{k-1}_{\tau^{-}_k} \geq \delta \big) \cap \big( Q^{u}_{\tau^{-}_k} < n_{\max} \big) \right\} } \right], 
\label{v-Swing-impulse-control}
\end{align}
where $u = (\tau_k)_{k \geq 1}$ is the sequence of $\F^{W}$-stopping times valued in $(0, T]$ 
and two additional state variables which both are controlled and discontinuous (c\`adl\`ag) processes are introduced:
\begin{itemize}
	\item $Q^{u}$ counts the number of exercise rights used before considered time
		$$
				Q^{u}_0 = 0, \quad Q^{u}_t = \sharp \left\{k \geq 1, \tau_k \leq t\right\}, \quad \forall t \geq 0.
		$$
	\item $\Theta^{u}_t := \Theta^{k}_t = \inf \left\{ t - \tau_k, \tau_k \leq t \right\}$ corresponds to the delay between $t$ and last exercise date
	$$
	\Theta^{k}_t = t - \tau_{k}, \forall \tau_{k} \leq t < \tau_{k+1}, \quad \Theta^{k}_{\tau_{k+1}} = 0, \quad \forall k \geq 0,
	$$
	where by convention $\Theta^{u}_0 = \Theta^{0}_0 = 0$.
\end{itemize}
Obviously, the problem \eqref{v-Swing} is equivalent to the impulse control problem \eqref{v-Swing-impulse-control}:
\begin{equation*}
\forall s \in \R, \quad \bar v^{(n_{\max})}(s) = v(s).
\end{equation*}
%

\noindent\textbf{\textsl{Penalized BSDE Associated to Swing Option Pricing Problem}} \\

\noindent Let us introduce the uncontrolled variable $(Q, \Theta)$ defined by
\begin{equation}
\begin{cases}
Q_t & = \ \ N_t, \hspace*{2cm} \forall t \geq 0, \\
\Theta_t & = \ \ t - \int^{t}_{0} \Theta_{s^{-}} dN_s,
\end{cases}
\label{fwd-Swing}
\end{equation}
where we recall that $N$ is a Poisson process with intensity $\lambda > 0$.
For a penalization coefficient $p > 0$, the penalized BSDE with jumps associated to problem \eqref{v-Swing-impulse-control} is
\begin{align}
Y^{p}_t & = \kappa(S_T, Q_{T^{-}}, \Theta_{T^{-}}) - \int^{T}_{t} r Y^{p}_u du + p \int^{T}_{t} (V^{p}_u + \kappa(S_{u}, Q_{u^{-}}, \Theta_{u^{-}}))^{+} \lambda du \nonumber \\
				& \hspace*{2.8cm} - \int^{T}_{t} Z^{p}_u dW_u - \int^{T}_{t} V^{p}_u dN_u, \quad \forall 0 \leq t \leq T.
\label{PBSDE-Swing}
\end{align}
in which the intervention gain is defined by
\begin{equation}
	\kappa (s, q, \theta) := \phi (s) \Ind_{ \left\{ \left( \theta \geq \delta \right) \cap \left( q \leq n_{\max} - 1 \right) \right\} }, \quad \forall (s, q, \theta) \in \R \times \N \times \R^{+}.
\label{kappa-Swing}
\end{equation}

\subsection{Numerical Valuation Algorithm} 
\label{subsec:algo}

The three-dimensional process $(S, Q, \Theta)$ can be exactly computed on the time grid $\pi$.
We shall denote by $(S^{\pi}, Q^\pi, \Theta^\pi)$ its version on $\pi$.
In particular, the pure jump processes $Q$ and $\Theta$ can be computed without approximation error, recall \eqref{fwd-Swing},
from a simulated trajectory of a Poisson process with intensity $\lambda$, see for example \cite{CT03}. 

As the driver of the penalized BSDE with jumps \eqref{PBSDE-Swing} does not depend on $Z^p$, it is sufficient to compute $\left( Y^{p, \pi}, V^{p, \pi} \right)$ on $\pi$, which can be made by the backward recursive scheme:
\begin{equation}
	\begin{cases}
		Y^{p, \pi}_{t_N} = \kappa (S^{\pi}_{t_N}, Q^{\pi}_{t_N}, \Theta^{\pi}_{t_{N}}) \\
		\forall t_n \in \pi, t_n < T : \\
		\hspace*{0.8cm} V^{p, \pi}_{t_n}
		= \frac{1}{\lambda \Delta t_{n+1}} \Esp \left[ Y^{p, \pi}_{t_{n+1}} \Delta \tilde N_{t_{n+1}} | \Fcurs_{t_n} \right]  \\
		\hspace*{0.8cm} Y^{p, \pi}_{t_n}
		 = \frac{1}{1 + r \Delta t_{n+1}} \left( \Esp \left[ Y^{p, \pi}_{t_{n+1}} | \Fcurs_{t_n} \right] \right. \\
		 \hspace*{3cm} \left. + \left[ p \left( V^{p, \pi}_{t_n} + \kappa(S^{\pi}_{t_n}, Q^{\pi}_{t_n}, \Theta^{\pi}_{t_n}) \right)^{+} - V^{p, \pi}_{t_n} \right] \lambda \Delta t_{n+1} \right)
	\end{cases}
	\label{bwd-scheme-Swing}
\end{equation}
and in our setting:
\begin{align}
\Esp \left[ \cdot | \Fcurs_{t_n} \right] = \Esp \left[ \cdot | (S^{\pi}_{t_n}, Q^{\pi}_{t_n}, \Theta^{\pi}_{t_{n}}) \right] \ .
\label{esp-cond}
\end{align}

\noindent\textbf{\textsl{Monte Carlo-based Resolution}} \\

\noindent We compute estimators of the conditional expectations \eqref{esp-cond} by a classical least squares Monte Carlo technique.
We use least squares regressions on adaptative local basis functions, see \cite{BW10}.
Since such an approach is only relevant for real-valued variables (the regression basis functions have compact support),
it cannot handle the integer-valued variable $Q^{\pi}$. \\

We thus deal with this variable explicitly.
Namely, we simulate $M \geq 1$ i.i.d. paths of $(S^{\pi}, Q^{\pi}, \Theta^{\pi})$
$$
(S^{\pi, (m)}, Q^{\pi, (m)}, \Theta^{\pi, (m)}), \quad \forall m \leq M,
$$
and at each time step $t_n < T$ of the backward recursion, the set of $M$ Monte Carlo samples is separated
in $n_{\max} + 1$ sub-sets corresponding to the samples on which $Q^{\pi}_{t_n} = 0, 1, \ldots, n_{\max - 1}$ and $Q^{\pi}_{t_n} \geq n_{\max}$.
Then, we just need to estimate $n_{\max}$ conditional expectation operators, namely
$$
\Esp \left[ \cdot | (S^{\pi}_{t_n}, Q^{\pi}_{t_n} = q, \Theta^{\pi}_{t_{n}}) \right], \quad \forall q \in \left\{ 0, 1, \ldots, n_{\max} - 1 \right\}
$$
by using the corresponding Monte Carlo samples, since $(Y^{p, \pi}_{t_n}, V^{p, \pi}_{t_n}) = (0, 0)$ when $Q^{\pi}_{t_n} \geq n_{\max}$, see Remark \ref{simplif-esp-cond}.

\begin{remark} Let us show that $Q^{\pi}_{t_n} \geq n_{\max}$ implies $Y^{p, \pi}_{t_n} = V^{p, \pi}_{t_n} = 0$. We have
\begin{eqnarray*}
& & \Esp \left[ Y^{p, \pi}_{t_{n+1}} \Delta \tilde N_{t_{n+1}} | (S^{\pi}_{t_n}, Q^{\pi}_{t_n} \geq n_{\max}, \Theta^{\pi}_{t_{n}}) \right] \\
& = & \Esp \Big[  \underbrace{\Esp \left[ Y^{p, \pi}_{t_{n+1}} |  Q^{\pi}_{t_n} \geq n_{\max} \right]}_{ = \ 0 } \Delta \tilde N_{t_{n+1}} |  (S^{\pi}_{t_n}, Q^{\pi}_{t_n} \geq n_{\max}, \Theta^{\pi}_{t_{n}})  \Big] \ = \ 0,
\end{eqnarray*}
by definition of $\kappa$ in \eqref{kappa-Swing} and \eqref{bwd-scheme-Swing}. In the same way
\begin{align*}
\Esp \left[ Y^{p, \pi}_{t_{n+1}} |  (S^{\pi}_{t_n}, Q^{\pi}_{t_n} \geq n_{\max}, \Theta^{\pi}_{t_{n}}) \right] 
\ = \ 0
\end{align*}
which allows to conlude.
\label{simplif-esp-cond}
\end{remark}


The Monte Carlo-based algorithm 
is then the following:
\begin{enumerate}[I.]
	\item Initialization:
	$$
	Y^{p, \pi, (m)}_{t_N} = \phi(S^{\pi, (m)}_{t_N}) \Ind_{ \left\{ \left( \Theta^{\pi, (m)}_{t_{N}} \geq \delta \right) \cap \left( Q^{\pi, (m)}_{t_{N}} \leq n_{\max} - 1 \right) \right\} }, \quad \forall m \leq M.
	$$
	\item Computation backward in time of $(V^{p, \pi, (m)}, Y^{p, \pi, (m)})$ on each sample $m\leq M$. \\
		For $n = N-1, \ldots, 0$, set
		\begin{align*}
			\Mcurs^q_{t_n} & := \left\{ m = 1, \ldots, M: Q^{\pi, (m)}_{t_n} = q \right\}, \quad \forall q \leq n_{\max} - 1, \\
			\Mcurs^{n_{\max}}_{t_n} & := \left\{ m = 1, \ldots, M: Q^{\pi, (m)}_{t_n} \geq n_{\max} \right\}.
		\end{align*}
		Then:
	\begin{enumerate}[1.]
		\item For any $m \in \Mcurs^{n_{\max}}_{t_n}$,
		$$
		V^{p, \pi, (m)}_{t_n} = Y^{p, \pi, (m)}_{t_n} = 0.
		$$
		\item Set $q := n_{\max} - 1$. \\
		\item If $q \geq 0$, for any $m \in \Mcurs^{q}_{t_n}$, the conditional expectations estimators
		\begin{align*}
		\varepsilon^{V, q, (m)}_{t_n} & \approx \Esp \left[ Y^{p, \pi}_{t_{n+1}} \Delta \tilde N_{t_{n+1}} | \ \left( S^{\pi, (m)}_{t_n}, \Theta^{\pi, (m)}_{t_{n}} \right) \right] \\
		\varepsilon^{Y, q, (m)}_{t_n} & \approx \Esp \left[ Y^{p, \pi}_{t_{n+1}} | \ \left( S^{\pi, (m)}_{t_n}, \Theta^{\pi, (m)}_{t_{n}} \right) \right]
		\end{align*}
		are approximated by least squares regression of $\left( Y^{p, \pi, (m)}_{t_{n+1}} \Delta \tilde N^{(m)}_{t_{n+1}} \right)_{m \in \Mcurs^{q}_{t_n}}$ and $\left( Y^{p, \pi, (m)}_{t_{n+1}} \right)_{m \in \Mcurs^{q}_{t_n}}$ respectively on
		$$
		\left( \psi_1( S^{\pi, (m)}_{t_n}, \Theta^{\pi, (m)}_{t_{n}} ), \ldots, \psi_b( S^{\pi, (m)}_{t_n}, \Theta^{\pi, (m)}_{t_{n}} ) \right)_{m \in \Mcurs^{q}_{t_n}}
		$$
		with $b := b^{q}_{t_n}$ basis functions $\{ \psi_1, \ldots, \psi_b \}$. Then,
		\begin{equation*}
	\begin{cases}
		V^{p, \pi, (m)}_{t_n} = \frac{1}{\lambda \Delta t_{n+1}} \ \varepsilon^{V, q, (m)}_{t_n}   \\
		Y^{p, \pi, (m)}_{t_n}
		 = \frac{1}{1 + r \Delta t_{n+1}} \left( \varepsilon^{Y, q, (m)}_{t_n} + \left[ p \left( V^{p, \pi, (m)}_{t_n} + \phi(S^{\pi, (m)}_{t_n}) \Ind_{ \left\{ \Theta^{\pi, (m)}_{t_{n}} \geq \delta \right\} } \right)^{+} \right. \right. \\
		 \hspace*{2.8cm} \left. \left. - \ V^{p, \pi, (m)}_{t_n} \right] \lambda \Delta t_{n+1} \right).
	\end{cases}
		\end{equation*}
		\item $q := q-1$ and go to 3.
	\end{enumerate}
	\item At time $t_0$ ($\Mcurs^0_{t_0} = \left\{ 1, \ldots, M \right\}$ and $\Mcurs^{q}_{t_0} = \emptyset, \forall q \geq 1$)
	the Swing option price estimator is given by $Y^{p, \pi}_{t_0}$ such that
	\begin{equation*}
	\begin{cases}
		V^{p, \pi}_{t_0} = \frac{1}{\lambda \Delta t_{1}} \frac{1}{M} \sum_{m=1}^M \left( Y^{p, \pi, (m)}_{t_{1}} \Delta \tilde N^{(m)}_{t_{1}} \right) \\
		Y^{p, \pi}_{t_0}
		 = \frac{1}{1 + r \Delta t_{1}} \left( \frac{1}{M} \sum_{m=1}^M Y^{p, \pi, (m)}_{t_{1}} + \left[ p ( V^{p, \pi}_{t_0} )^{+} - V^{p, \pi}_{t_0} \right] \lambda \Delta t_{1} \right).
	\end{cases}
		\end{equation*}
	
\end{enumerate}

Let us highlight some features of the above-presented Monte Carlo procedure. 
At each backward induction date $t_n < T$, we have to estimate in worst cases $2 \times n_{\max}$ conditional expectations, performed on each subset $\Mcurs^q_{t_n}, q \leq n_{\max} - 1$. When $n_{\max}$ increases, much more Monte Carlo samples are needed as each least squares regression requires a sufficient number of samples.

In addition, the number of local basis functions $b^{q}_{t_n}$ has to be adapted to the number of Monte Carlo samples used for the least squares regression, namely $\card(\Mcurs^q_{t_n})$.
We thus introduce a dynamic choice for $b^{q}_{t_n}$: it is fixed proportionally to $\card(\Mcurs^q_{t_n})$ for any $q \leq n_{\max}-1$ and $t_n<T$.

\begin{remark}[Statistical error of our method] 
A control of the statistical error introduced by the least squares Monte Carlo approach
is provided in Gobet et al. \cite{GL06} (see \cite{Lem05} for further details).
By extension, this applies to BSDE with jumps (see \cite{Eli08}) to control the error
on the jump component $V^{p, \pi}$ and thus ensures that the least squares Monte Carlo error tends to $0$
as the number of samples $M$ and the number of basis functions $b$ tends to $+ \infty$.

\end{remark}

\noindent\textbf{A Benchmark Method Based on Iteration} \\

\noindent The classical method to value such a Swing option, recall formulation \eqref{v-Swing},
is based on an iteration over the number of exercise rights, see for example \cite{CT08}.
The dynamic programming principle provides a direct link between the solution $v^{(j)}$
to the same problem as \eqref{v-Swing} but with at most $j \leq n_{\max}$ exercise rights
and the solution $v^{(j-1)}$ with at most $(j-1)$ exercise rights.

The value of the Swing option with $0$ exercise right is obviously zero $v^{(0)} = 0$.
Then, we compute the sequence of values of Swing options with $j$ exercise rights $v^{(j)}$, $j = 1, \ldots, n_{\max}$ according the backward recursion scheme:
\begin{equation*}
	\begin{cases}
		v^{(j)}(t_N, s) = \phi(s) \\
		\forall t_n \in \pi, T - \delta < t_n < T : \\
		\hspace*{0.8cm} v^{(j)}(t_n, s) = \max \left\{ \phi(s) \ ; \ e^{-r \Delta t_{n+1}} \Esp^{(t_n, s)} \left[ v^{(j)}(t_{n+1}, S^{\pi}_{t_{n+1}}) \right] \right\} \\
		\forall t_n \in \pi, t_n \leq T - \delta : \\
		\hspace*{0.8cm} v^{(j)}(t_n, s) = \max \left\{ \phi(s) + e^{-r \delta} \Esp^{(t_n, s)} \left[ v^{(j - 1)} (t_n + \delta, S^{\pi}_{t_n + \delta}) \right] ; \right. \\
		\hspace*{3.3cm} \left. e^{-r \Delta t_{n+1}} \Esp^{(t_n, s)} \left[ v^{(j)}(t_{n+1}, S^{\pi}_{t_{n+1}}) \right] \right\}
	\end{cases}
\end{equation*}
where $\Esp^{(t_n, s)} \left[ \cdot \right] := \Esp \left[ \cdot | S^{\pi}_{t_n} = s \right]$.
We use the same least squares Monte Carlo regression-based method for approximating the conditional expectations operators as above.

\subsection{Pricing Results}

We consider put options with maturity $T = 1$ year and a strike price $K = 100$.
The Black and Scholes parameters, see \eqref{prix-BS}, are $r = 0.05$, $\sigma = 0.3$ and $s = 100$.

\subsubsection{Special Case of American Options: $n_{\max} = 1$}

In the single-exercise case, the additional variable $\Theta$ disappears (there is no delay constraint).
This helps simplify the Monte Carlo procedure described in Paragraph \ref{subsec:algo}.
In particular, the algorithm implies only two sequences of samples subsets, whether one jump of $N$ occurs before $T$ on the considered path or not: that is
$((\Mcurs^1)_{t_n})_{n \leq N-1}$ and $((\Mcurs^0)_{t_n})_{n \leq N-1}$. \\

In our numerical experiments, we find out that 
increasing too much $\lambda$ makes the variance of the Monte Carlo procedure explode.
It would be necessary to increase the number of Monte Carlo samples, which leads to prohibitive computational times
(each pricing result presented below was obtained after a computation between 6 and 8 hours).
For the same reason (exploding behavior of the penalized BSDE driver), we restrict our numerical experiments to penalization parameters $\leq 5$. \\

The benchmark price for the American put option is $9.88$ (by a binomial approach or classical Monte Carlo). We report
in Table \ref{tab:put-US} the price given by our method when varying $\lambda$ and
the number of time steps $N$ for a penalization parameter equal to $5$. We used 20 million of Monte Carlo paths.

\begin{table}  
\begin{center}
\begin{tabular}{p{1.5cm}p{1.5cm}p{1.5cm}p{1.5cm}p{1.5cm}p{1.5cm}}
\hline\noalign{\smallskip}
\backslashbox{$\lambda$}{$N$} & 20 & 40 & 80 & 160 & 320 \\
\hline
\hline
	3 & 9.89 & 9.92	& 9.95	& 9.94	& 9.83 \\
	4 & 9.92 & 9.96	& 9.99 &	9.97	& 9.83 \\
	5 & 9.95 & 9.99	& 10.02	& 9.98	& 9.76 \\
\noalign{\smallskip}\hline\noalign{\smallskip}
\end{tabular}
\end{center}
\caption{Approximate prices of an American option with $p=5$}
\label{tab:put-US}    
\end{table}

In all the experiments that we performed in this simple case,
we numerically observed that the limiting prices of our method
(with respect to $N$) are below the benchmark value: this is due to penalization.

\subsubsection{Swing Options with $n_{\max} = 2$}

We consider time delays $\delta = \frac{1}{10}, \frac{2}{10}, \frac{3}{10}$. \\

The benchmark prices for the Swing put option with $2$ exercise rights are $19.27$, $18.77$ and $18.21$ respectively
(computed with the method described in Paragraph \ref{subsec:algo}, $N = 200$ time steps and $M = 5$ million of Monte Carlo paths).
We report in Figures 1, 2 and 3 the corresponding approximate prices
when varying $\lambda$ and $N$ for a penalization parameter equal to $5$ and $10$ (we used $40$ million of Monte Carlo paths
and $N = 20, 40, 80, 160, 320, 640$).

\begin{figure}[h]
	\centering
	\begin{minipage}{0.5\linewidth}
      \centering \includegraphics[width=7cm]{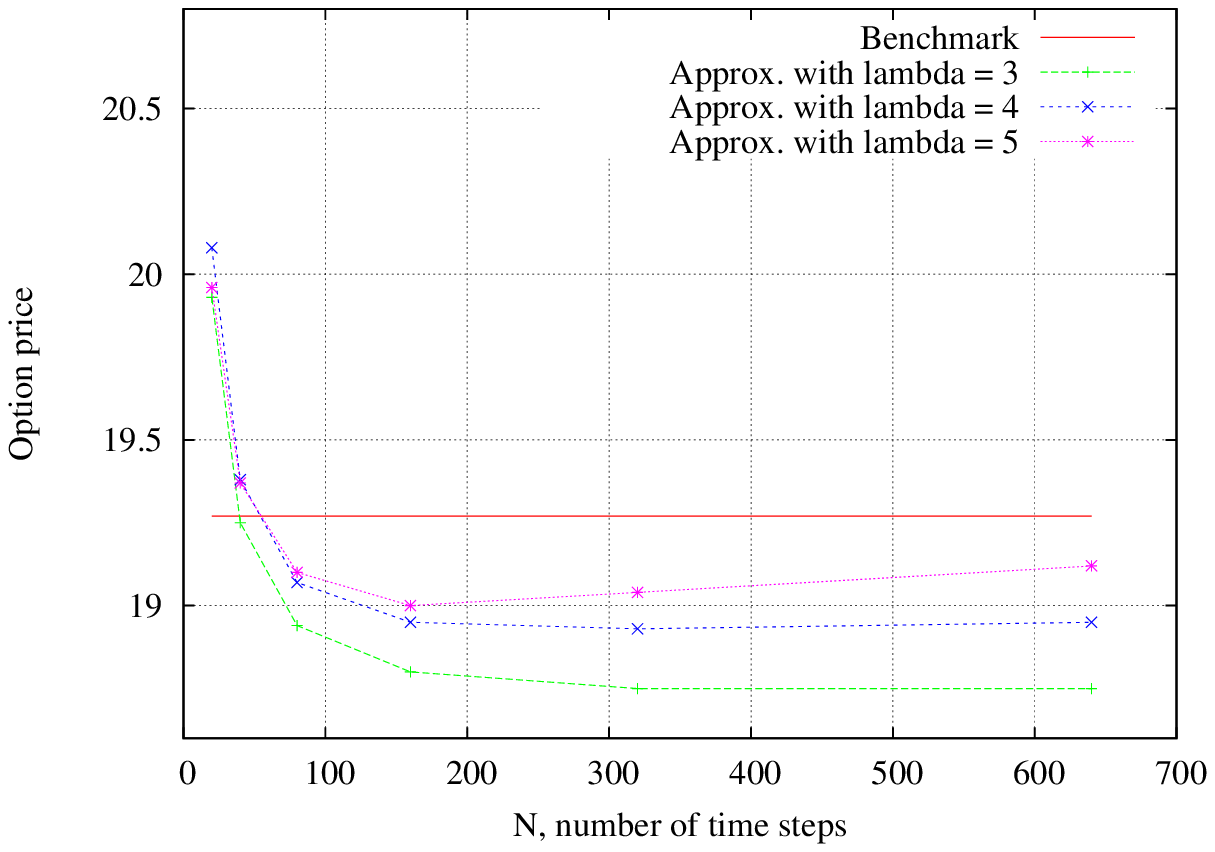}
   \end{minipage}\hfill
   \begin{minipage}{0.5\linewidth}
      \centering \includegraphics[width=7cm]{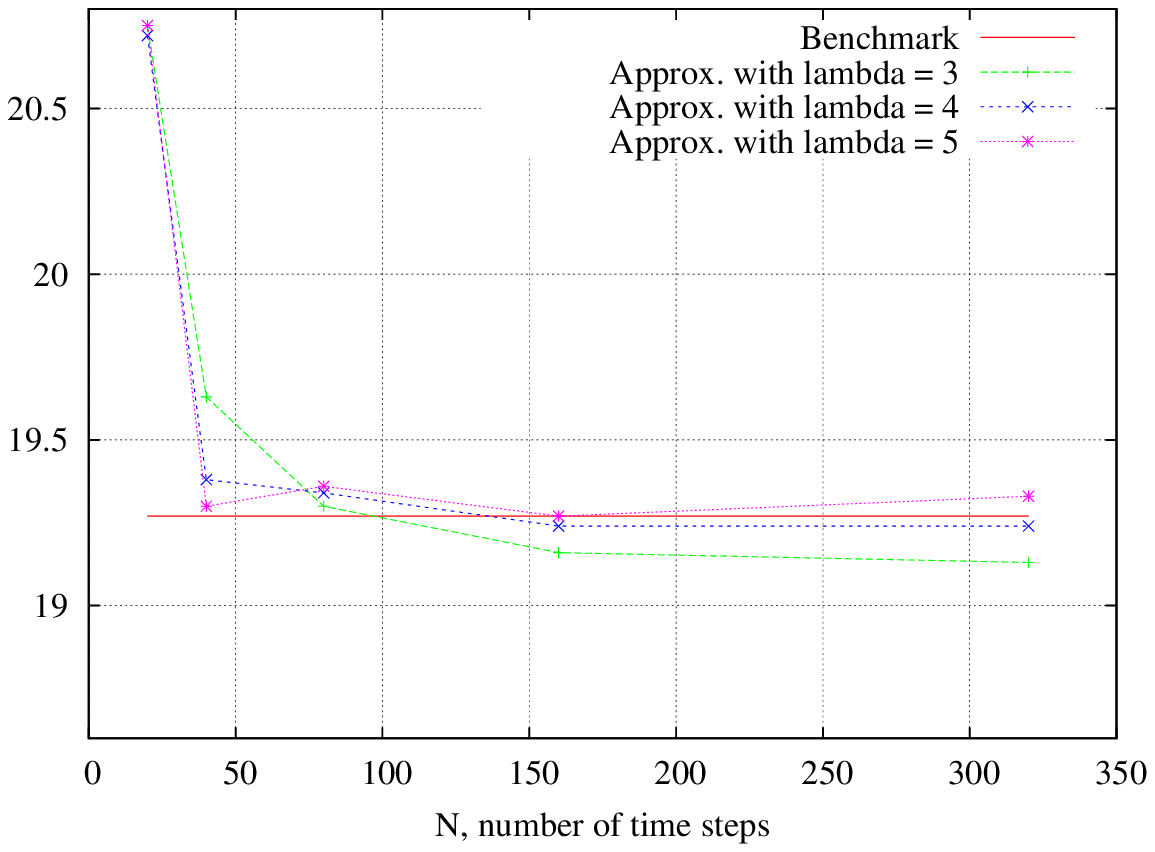}
   \end{minipage} \hfill
	\centering{\caption{Approximate prices of a Swing option with $2$ exercise rights}and $\delta = \frac{1}{10}$, with $p=5$ (left) and $p=10$ (right)}
	\label{fig:put-swing-1}
\end{figure}

\begin{figure}[h]
	\centering
	\begin{minipage}{0.5\linewidth}
      \centering \includegraphics[width=7cm]{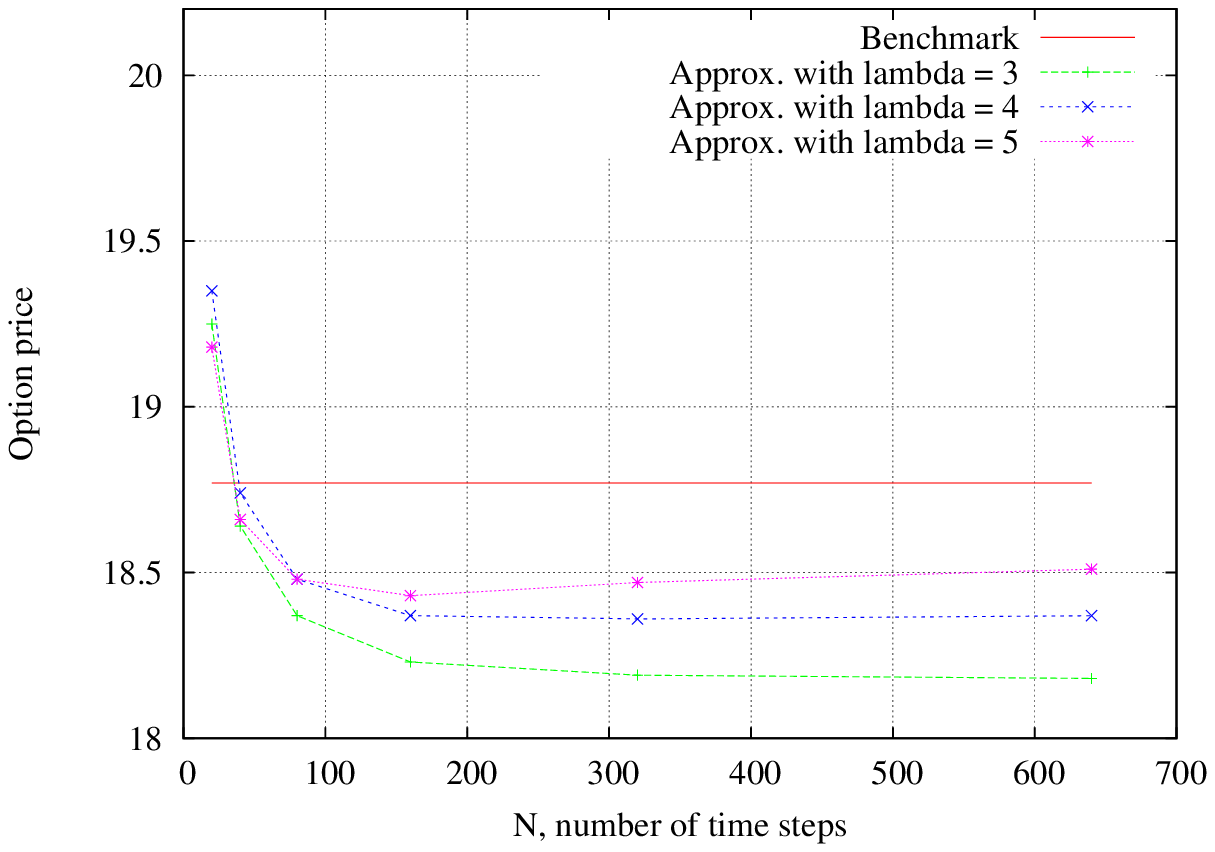}
   \end{minipage}\hfill
   \begin{minipage}{0.5\linewidth}
      \centering \includegraphics[width=7cm]{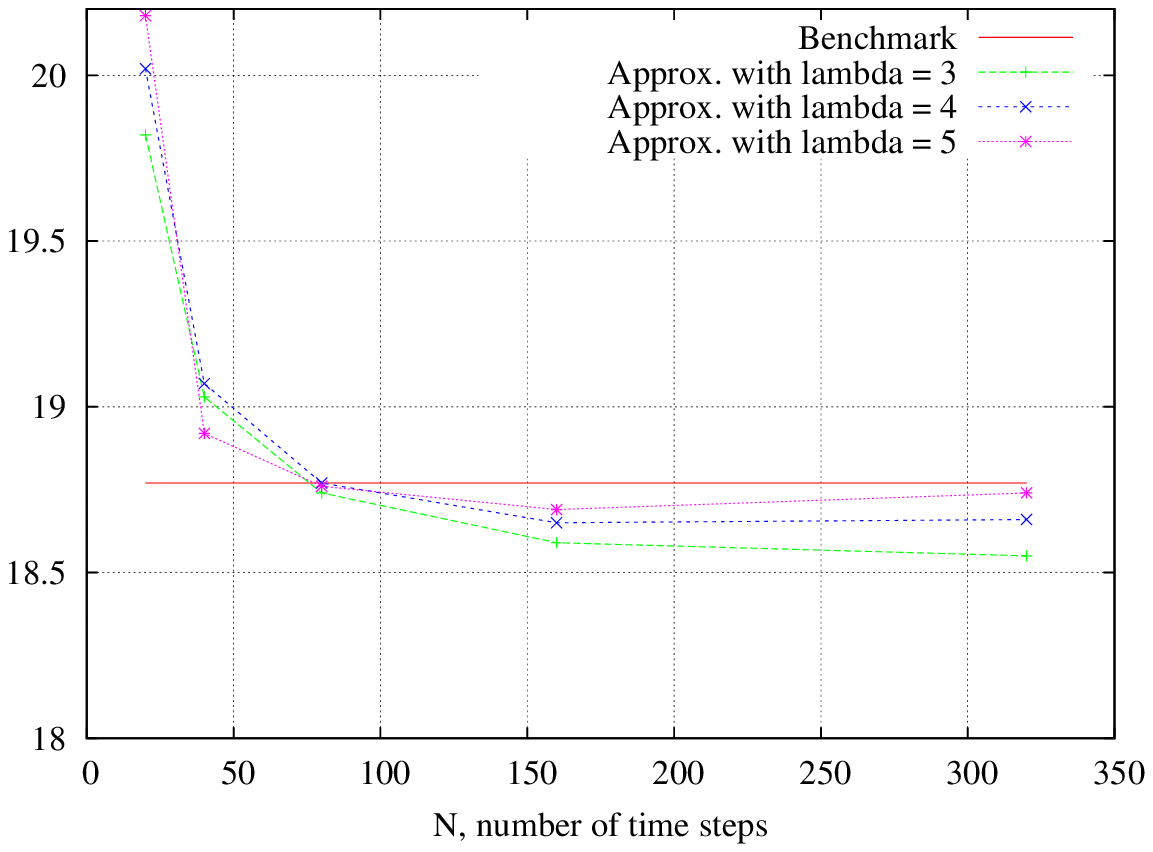}
   \end{minipage} \hfill
	\centering{\caption{Approximate prices of a Swing option with $2$ exercise rights}and $\delta = \frac{2}{10}$, with $p=5$ (left) and $p=10$ (right)}
	\label{fig:put-swing-2}
\end{figure}

\begin{figure}[h]
	\centering
	\begin{minipage}{0.5\linewidth}
      \centering \includegraphics[width=7cm]{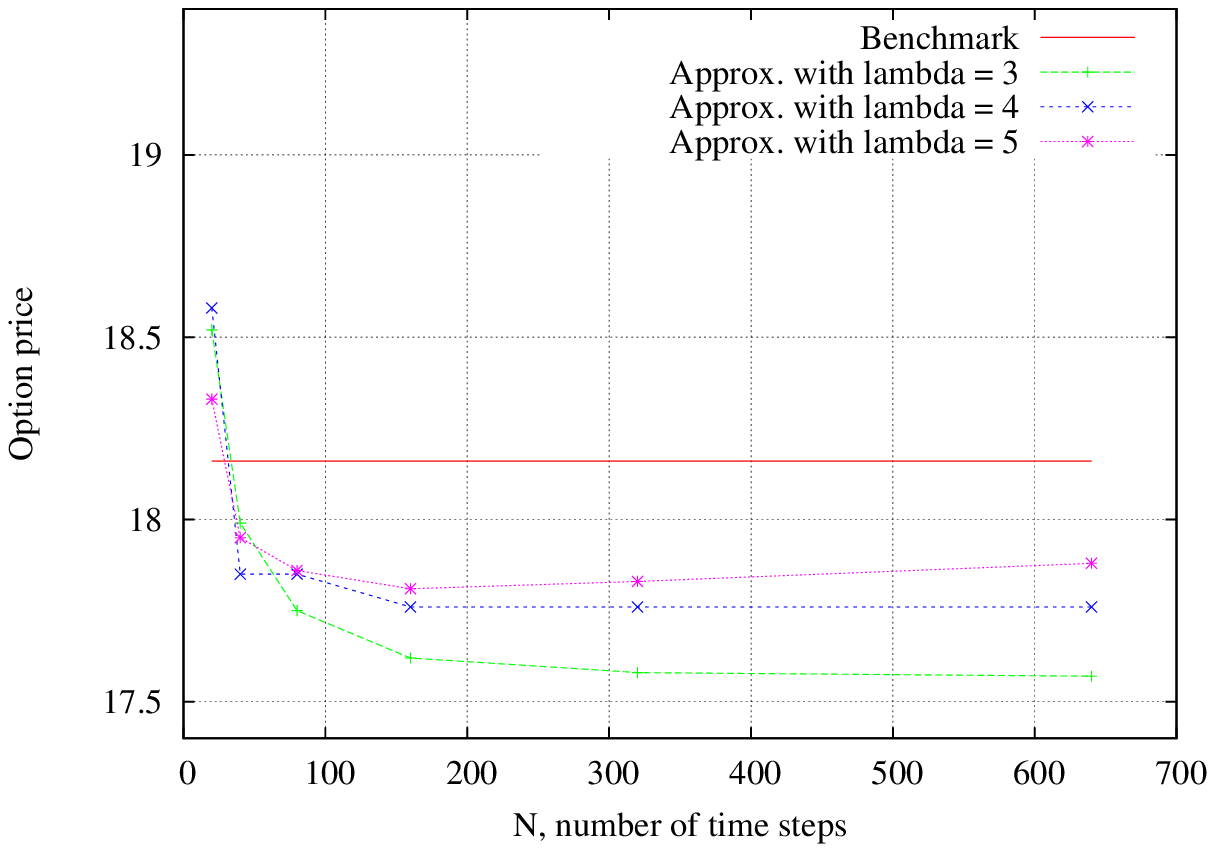}
   \end{minipage}\hfill
   \begin{minipage}{0.5\linewidth}
      \centering \includegraphics[width=7cm]{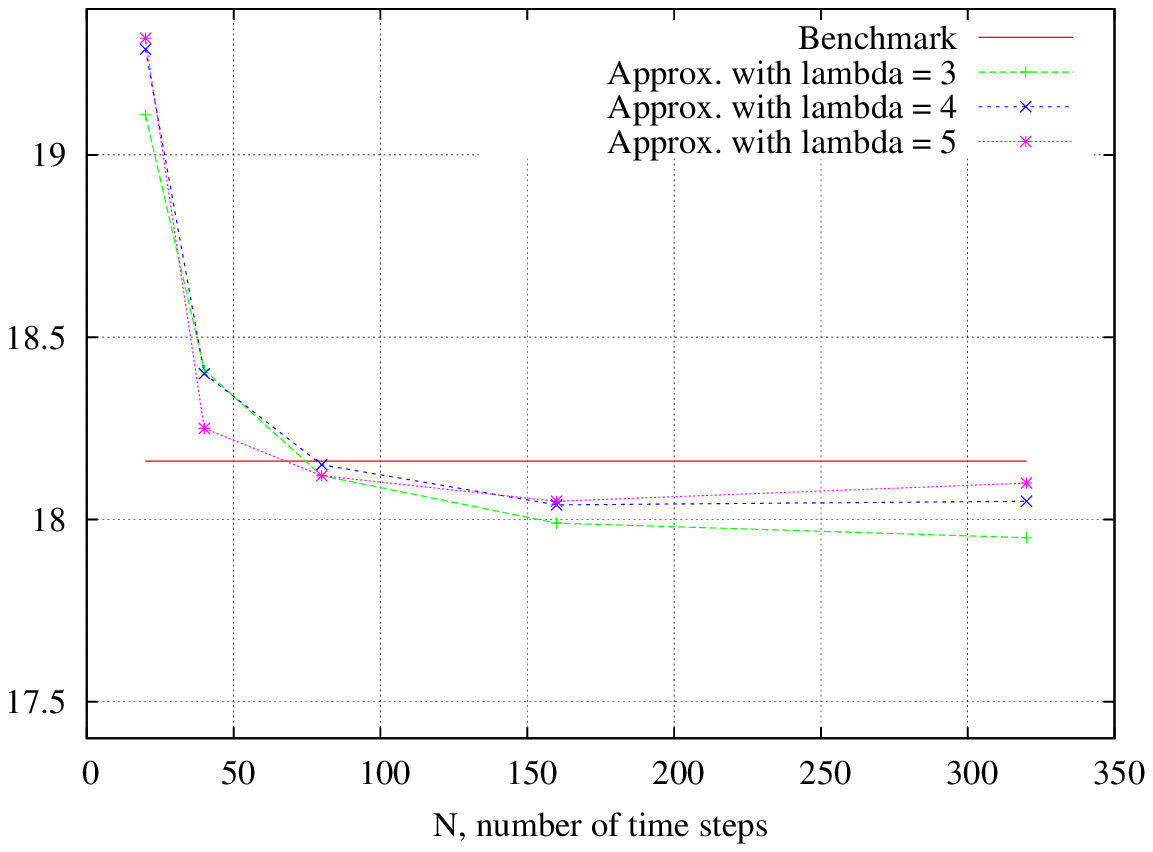}
   \end{minipage} \hfill
	\centering{\caption{Approximate prices of a Swing option with $2$ exercise rights}and $\delta = \frac{3}{10}$, with $p=5$ (left) and $p=10$ (right)}
	\label{fig:put-swing-3}
\end{figure}

For each considered value of $\lambda$, we retrieve a convergence in the number of time steps $N$
of our method. As $p=5$, approximate prices are converged from $N = 160$, so that we restrict ourselves to $N \leq 320$ time steps as $p=10$. 
The limiting values are still below the benchmark 
but accurate option prices (relative error less than $1\%$) are obtained with a penalization coefficient $p$ equal to $10$ and $N=160$.
See also Table \ref{tab:put-swing-CV} in which the (signed) relative error to the benchmark is given in brackets. 
Besides, we observe a monotone convergence in $\lambda$ of our approximate method. 

\begin{table}
\begin{center}
\begin{tabular}{p{1.5cm}p{1.5cm}p{2.5cm}p{2.5cm}p{2.5cm}}
\hline\noalign{\smallskip}
				& \backslashbox{$p$}{$\lambda$} & $3$ & $4$ & $5$ \\
\hline
\hline
	\multirow{2}{*}{$\delta = \frac{1}{10}$} & $5$ & 18.80	(-2.44\%) &	18.95	(-1.66\%) &	19.00	(-1.40\%) \\
								 & $10$ & 19.16	(-0.57\%) &	19.24	(-0.16\%) &	19.27	(0.00\%) \\
\noalign{\smallskip}\hline\noalign{\smallskip}
	\multirow{2}{*}{$\delta = \frac{2}{10}$} & $5$ & 18.23	(-2.88\%) &	18.37	(-2.13\%) &	18.43	(-1.81\%) \\
								 & $10$ & 18.59	(-0.96\%) &	18.65	(-0.64\%) &	18.69	(-0.43\%) \\	
\noalign{\smallskip}\hline\noalign{\smallskip}
	\multirow{2}{*}{$\delta = \frac{3}{10}$} & $5$ & 17.62	(-2.97\%) &	17.76	(-2.20\%) &	17.81	(-1.93\%) \\
								 & $10$ & 17.99	(-0.94\%) &	18.04	(-0.66\%) &	18.05	(-0.61\%) \\
\noalign{\smallskip}\hline\noalign{\smallskip}
\end{tabular}
\end{center}
\caption{Prices of a Swing option with $2$ exercise rights}\centering{(limiting values with $N=160$)}
\label{tab:put-swing-CV}      
\end{table}

We should point out that fine-tuning the parameters of the algorithm is difficult. As already mentioned, 
since the number of Monte Carlo paths is different in each set of sample paths $\Mcurs^q_{t_n}, q = 0, 1, 2$, the number of basis functions used for the least squares regressions has to be dynamically adapted. And when increasing much more the jump intensity $\lambda$, more Monte Carlo samples would be necessary. 

For such a Swing option, the running time is much longer because the conditional expectations are computed by regression with respect to the bidimensional state variable $(S^{\pi}, \Theta^{\pi})$. The computation of one option price takes at least 15 hours in above cases (when $N \geq 80$). In comparison, the benchmark method takes less than 5 minutes. Besides, the complexity of our method increases with $n_{\max}$, leading to untractable computational times for bigger values of $n_{\max}$, see Remark \ref{rem-increase-n-max}. \\

On this particular case of Swing options valuation, it seems that our method is less competitive 
than the classical approach. 
This is without any doubt due to the strong degeneracy of such a problem in our impulse control context:
the valuation problem is $3$-dimensional and involves an additional integer-valued state variable $Q$
representing the number of exercise rights used at any considered time.

However, our method works and the numerical results that we obtain are consistent
with the theoretical convergence rate given in Theorem \ref{theo:cv-rate-global}.
One can expect that our method would work better on less degenerate problems.

\begin{remark}[Dealing with more exercise rights]
\label{rem-increase-n-max}
The computational time of our method intuitively increases \textit{linearly} with the number of exercise rights $n_{\max}$.
Indeed, at each time step of the backward induction procedure, the number of conditional expectation estimations is proportional to $n_{\max}$.
Besides, when multiplying by $2$ the number of exercise rights, it would require, at least, a double number of Monte Carlo samples for a same accuracy of the
computation of conditional expectation estimators.

Let us mention that the computational time of the benchmark method using iteration 
increases linearly as function of the the maximal number of exercise rights as well.
\end{remark}

%

%
%
%

\begin{thebibliography}{99.}%
%
%


%
	\bibitem{Ber11} M. Bernhart: 
					Modelization and valuation methods of gas contracts, see \underline{Part II}.
					PhD Thesis, Universit\'e Paris VII Denis-Diderot, Preprint (2010)
				
	\bibitem{BouC07} B. Bouchard and J-F. Chassagneux:
					Discrete-time approximation for continuously and discretely reflected BSDEs.
					Stochastic Processes and their Applications, \textbf{118}(12), 2269-2293 (2008)
					
	\bibitem{BouE08} B. Bouchard and R. Elie:
					Discrete-time approximation of decoupled forward-backward SDE with jumps.
					Stochastic Processes and their Applications, \textbf{118}(1), 53-75 (2008)
								
	\bibitem{BL84} A. Bensoussan and J.-L. Lions:
					Impulse control and quasi-variational inequalities.
					Gauthier-Villars (1984)
					
%
	\bibitem{BT04} B. Bouchard and N. Touzi:
	 				Discrete-Time Approximation and Monte-Carlo Simulation of Backward Stochastic Differential Equations.
	 				Stochastic Processes and their applications, \textbf{111}(2), 175-206 (2004)
	 				
	\bibitem{BW10} B. Bouchard and X. Warin:
					Monte-Carlo valorisation of American options: facts and new algorithms to improve existing methods.
					Preprint (2010)
					
					
	\bibitem{BP03a} V. Bally and G. Pag\`es:
					Error analysis of the optimal quantization algorithm for obstacle problems.
					Stochastic Processes and Their Applications, \textbf{106}(1), 1-40 (2003)
					
					
					
	\bibitem{CL07b} R. Carmona and M. Ludkovski:
					Pricing Asset Scheduling Flexibility Using Optimal Switching.
					Applied Mathematical Finance, \textbf{15}(6), 405-447 (2008)
					
	\bibitem{CT08} R. Carmona and N. Touzi:
					Optimal multiple stopping and valuation of Swing options.
					Mathematical Finance, \textbf{18}(2), 239-268 (2008)
					
				
	
	\bibitem{CE10} J-F. Chassagneux, R. Elie and I. Kharroubi:
					Discrete-time Approximation of Multidimensional BSDEs with oblique reflections.
					Preprint (2010)
					
	\bibitem{CT03} R. Cont and P. Tankov:
					Financial modelling with jump processes.
					Chapman \& Hall/CRC Pres (2003)
					
	\bibitem{CO00} J.-P. Chancelier, B. \O ksendal and A. Sulem:
					Combined stochastic control and optimal stopping, and application to numerical approximation of combined stochastic and impulse control.
					Tr. Mat. Inst. Steklova, \textbf{237}, 149-172 (2000)
	
	\bibitem{DH10} B. Djehiche, S. Hamad\`ene and I. Hdhiri:
					Stochastic Impulse Control of Non-Markovian Process.
					Applied Mathematics and Optimisation, \textbf{61}(1), 1-26 (2010)
					
	\bibitem{DHP09} B. Djehiche, S. Hamad\`ene and A. Popier:
					The Finite Horizon Optimal Multiple Switching Problem.
					SIAM Journal on Control and Optimisation, \textbf{48}(4), 2751-2770 (2010)
					
					

	\bibitem{Eli08} R. Elie:
					{Contrôle stochastique et m\'ethodes num\'eriques en finance math\'ematique}.
					PhD Thesis, Universit\'e Paris Dauphine (2008)
					
	\bibitem{EK10} R. Elie and I. Kharroubi:
					{Constrained BSDEs with jumps: Application to optimal switching}.
					Preprint (2010)
					
	\bibitem{GL05} E. Gobet, J.P. Lemor, X. Warin:
   				{A regression-based Monte Carlo method to solve backward stochastic differential equations}.
   				Annals of Applied Probability, \textbf{15}(3), 2172-2202 (2005)
          
  \bibitem{GL06} E. Gobet, J.P. Lemor and X. Warin:
					{Rate of convergence of empirical regression method for solving generalized BSDE}.
					Bernoulli, \textbf{12}(5), 889-916 (2006)
					
	\bibitem{HJ07} S. Hamad\`ene and M. Jeanblanc:
					{On the Starting and Stopping Problem: Application in reversible investments}.
					Mathematics of Operations Research, \textbf{32}(1), 182-192 (2007)
					
	\bibitem{HT07} Y. Hu and S. Tang:
					{Multi-dimensional BSDE with oblique reflection and optimal switching}.
					Probability Theory and Related Fields, \textbf{147}(1-2), 89-121 (2008)

	\bibitem{KM10} I. Kharroubi, J. Ma, H. Pham and J. Zhang:
					{Backward SDEs with constrained jumps and Quasi-Variational Inequalities}.
					Annals of Probability, \textbf{38}(2), 794-840 (2010)
					
	\bibitem{LMP07} V. Ly Vath, M. Mnif et H. Pham:
					{A model of portfolio selection under liquidity risk and price impact}.
					Finance and Stochastics, \textbf{11}(1), 51-90 (2007)
					
%
%
	
	\bibitem{Lem05} J-P. Lemor:
					{Approximation par projections et simulations de Monte-Carlo des \'equations diff\'erentielles stochastiques r\'etrogrades}.
					PhD Thesis, Ecole Polytechnique (2005)
					
	\bibitem{Lud05} M. Ludkovski:
					{Optimal Switching with Applications to Energy Tolling Agreements}.
					PhD Thesis, Princeton University (2005)
					
					
	\bibitem{OS07} B. \O ksendal and A. Sulem:
					{Applied stochastic control of jump diffusions}.
					Universitext, Springer Verlag (2006)
					
	\bibitem{Por08} A. Porchet:
					{Problems of Valuation and Organization in Energy Markets}.
					PhD Thesis, Universit\'e Paris Dauphine (2008)
	
	\bibitem{Sey09} R. C. Seydel:
					{Impulse Control for Jump-Diffusions: Viscosity Solutions of Quasi-Variational Inequalities and Applications in Bank Risk Management}.
					PhD Thesis, Leipzig Universität (2009)
					

%
\end{thebibliography}
%

\end{document}